\documentclass[lotsofwhite]{patmorin}
\usepackage{fullpage,graphicx,graphics,url,amsmath}
 
\usepackage{amsthm}

\newcommand{\comment}[1]{}

\newcommand{\seclabel}[1]{\label{sec:#1}}
\newcommand{\Secref}[1]{Section~\ref{sec:#1}}
\newcommand{\secref}[1]{\mbox{Section~\ref{sec:#1}}}

\newcommand{\figlabel}[1]{\label{fig:#1}}

\newcommand{\figref}[1]{\mbox{Figure~\ref{fig:#1}}}

\newcommand{\eqlabelx}[1]{\label{eq:#1}}
\newcommand{\eqrefx}[1]{(\ref{eq:#1})}

\newtheorem{thm}{Theorem}{\bfseries}{\itshape}
\newcommand{\thmlabel}[1]{\label{thm:#1}}
\newcommand{\thmref}[1]{Theorem~\ref{thm:#1}}

\newtheorem{lem}{Lemma}{\bfseries}{\itshape}
\newcommand{\lemlabel}[1]{\label{lem:#1}}
\newcommand{\lemref}[1]{Lemma~\ref{lem:#1}}

{\bfseries}{\itshape}

{\bfseries}{\itshape}

{\bfseries}{\itshape}

\newtheorem{assumption}{Assumption}{\bfseries}{\rm}

\newcommand{\etal}{\emph{et al.}}

\newcommand{\ceil}[1]{{\left\lceil #1 \right\rceil}}
\newcommand{\floor}[1]{{\left\lfloor #1 \right\rfloor}}
\newcommand{\R}{\mathbb{R}}
\newcommand{\N}{\mathbb{N}}

\usepackage{marvosym}

\newcommand{\boundary}{\partial}
\newcommand{\closure}{\mathrm{clo}}
\newcommand{\interior}{\mathrm{int}}
\newcommand{\z}[1]{{\hat{#1}}}
\newcommand{\depth}{\mathrm{depth}}
\newcommand{\prnt}{\mathrm{parent}}

\title{\MakeUppercase{Entropy, Triangulation, and Point Location} \\
	\MakeUppercase{in Planar Subdivisions}}

\author{%
  S\'ebastien Collette%
    \thanks{Charg\'e de Recherches du F.R.S.-FNRS.} \\ 
    \textit{Universit\'e Libre de Bruxelles}
  \and Vida Dujmovi\'c%
    \thanks{Research partially supported by NSERC and the Ontario
            Ministry of Research and Innovation.} \\ 
    \textit{Carleton University}
  \and John Iacono%
    \thanks{Research partially supported by
            NSF grants OISE-0334653 and CCF-0430849 
            and by an Alfred P. Sloan research fellowship.} \\ 
    \textit{Polytechnic Institute of NYU}
  \and Stefan Langerman%
    \thanks{Ma\^{\i}tre de recherches du F.R.S.-FNRS.} \\ 
    \textit{Universit\'e Libre de Bruxelles}
  \and Pat Morin%
    \thanks{This research began
while the author was a visiting researcher at the Universit\'e Libre de
Bruxelles and was completed while
 the author was a visiting researcher at National ICT 
Australia and the University of
Sydney. Research partially supported by NSERC, FNRS, the Ontario
            Ministry of Research and Innovation, Canada Foundation 
            for Innovation, The University of Sydney, and National 
            ICT Australia.} \\ 
    \textit{Carleton University}
}

\begin{document}
\maketitle

\begin{abstract} 
A data structure is presented for point location in connected planar
subdivisions when the distribution of queries is known in advance.
The data structure has an expected query time that is within a
constant factor of optimal.  More specifically, an algorithm is
presented that preprocesses a connected planar subdivision $G$ of size
$n$ and a query distribution $D$ to produce a point location data
structure for $G$. The expected number of point-line comparisons
performed by this data structure, when the queries are distributed
according to $D$, is $\tilde H + O(\tilde H^{2/3}+1)$ where $\tilde
H=\tilde H(G,D)$ is a lower bound on the expected number of point-line
comparisons performed by any linear decision tree for point location
in $G$ under the query distribution $D$.  The preprocessing algorithm
runs in $O(n\log n)$ time and produces a data structure of size
$O(n)$.  These results are obtained by creating a Steiner
triangulation of $G$ that has near-minimum entropy.
\end{abstract}

%%%%%%%%%%%%%%%%%%%%%%%%%%%%%%%%%%%%%%%%%%%%%%%%%%%%%%%%%%%%%%%%%%%
\section{Introduction}
\seclabel{intro}

The planar point location problem is the classic search problem in
computational geometry. Given a planar subdivision $G$,\footnote{A
\emph{planar subdivision} is a partitioning of the plane into points
(called \emph{vertices}), open line segments (call \emph{edges}), and
maximal connected 2-dimensional regions (called \emph{faces}).} the
planar point location problem asks us to construct a data structure so
that, for any query point $q$, we can quickly determine which face of
$G$ contains $q$.\footnote{In the degenerate case where $q$ is a
vertex or contained in an edge of $G$ any face incident on that
vertex/edge may be returned as an answer.}

The history of the planar point location problem parallels, in many
ways, the history of binary search trees.  After a few initial attempts
\cite{dl76,lp77,p81}, asymptotically optimal (and quite different)
linear-space $O(\log n)$ query time solutions to the planar point
location problem were obtained by Kirkpatrick \cite{k83}, Sarnak and
Tarjan \cite{st86}, and Edelsbrunner \etal\ \cite{egs86} in the mid
1980s.  These results were based on hierarchical simplification, data
structural persistence, and fractional cascading, respectively.  All
three of these techniques have subsequently found many other
applications.  An elegant randomized solution, combining aspects of
all three previous solutions, was later given by Mulmuley \cite{m90},
and uses randomized incremental construction, a technique that has
since become pervasive in computational geometry
\cite[Section~9.5]{bcko08}.  Preparata \cite{p90} gives a
comprehensive survey of the results of this era.

In the 1990s, several authors became interested in determining the
exact constants achievable in the query time.  Goodrich \etal\
\cite{gor97} gave a linear-size data structure that, for any query,
requires at most $2\log n + o(\log n)$ point-line comparisons and
conjectured that this query time was optimal for linear-space data
structures.\footnote{Here and throughout, logarithms are implicitly base 2
unless otherwise specified.} The following year, Adamy and Seidel
\cite{as98} gave a linear-space data structure that answers queries
using $\log n + 2\sqrt{\log n} + O(\log\log n)$ point-line comparisons
and showed that this result is optimal up to the third term.

Still not done with the problem, several authors considered the point
location problem under various assumptions about the query
distribution.  All these solutions compare the expected query time to
the \emph{entropy bound};  in a planar subdivision $G$ with $m$ faces
$F_1,\ldots,F_m$, if $\Pr(F_i)$ is the probability that $q$ is
contained in $F_i$, then no algorithm that makes only binary decisions
can answer queries using an expected number of decisions that is fewer
than 
\begin{equation}
    H = H(G,D) = \sum_{i=1}^m \Pr(F_i)\log(1/\Pr(F_i)) \enspace . 
	\eqlabelx{entropy-face}
\end{equation}

In the previous results on planar point location, none of the query
times are affected significantly by the structure of $G$;  they hold
for arbitrary planar subdivisions.  However, when studying point
location under a distribution assumption the problem becomes more
complicated and the results become more specific.  A \emph{connected
subdivision} is a planar subdivision whose underlying (vertex and
edge) graph is connected.  A \emph{convex subdivision} is a planar
subdivision whose faces are all convex polygons, except the outer
face, which is the complement of a convex polygon.  A
\emph{triangulation} is a convex subdivision in which each face has at
most 3 edges on its boundary.

Note that, if every face of $G$ has a constant number of sides, then
$G$ can be augmented, by the addition of extra edges, so that it is a
triangulation without increasing \eqrefx{entropy-face} by more than a
constant.  Similarly, in a convex subdivision $G$ where the query
distribution $D$ is uniform within each face of $G$, the faces of the
subdivision can be triangulated without increasing the entropy by more
than a constant \cite{amm00}. Thus, in the following we will simply
refer to results about triangulations where it is understood that
these also imply the same result for planar subdivisions with faces of
constant size or convex subdivisions when the query distribution is
uniform within each face.

Arya \etal\ \cite{acmr00} gave two results for the case where the
query point $p$ is chosen from a known distribution where the $x$ and
$y$ coordinates of $p$ are chosen independently and $G$ is a convex
subdivision.  They gave an $O(n)$ space data structure for which the
expected number of point-line comparisons is at most $4H+O(1)$ and an
$O(n^2)$ space data structure for which the expected number of
point-line comparisons is at most $2H+O(1)$.  The assumption about the
independence of the $x$ and $y$ coordinates of $p$ is crucial to the
these results.

For arbitrary distributions that are known in advance, several results
exist.  Iacono \cite{i01,i04} showed that, for triangulations, a
simple variant of Kirkpatrick's original point location structure
gives a linear space, $O(H+1)$ expected query time data structure.
Simultaneously, and independently, Arya \etal\ \cite{amm01b} showed
that a variant of Mulmuley's randomized data structure also achieves
$O(H+1)$ expected query time.  A sequence of papers by Arya \etal\
\cite{amm00,amm01a,ammw07} has recently culminated in an $O(n)$ space
data structure for point location in triangulations with query time
$H+O(H^{1/2} + 1)$ \cite{ammw07}.

In the current paper, we show that, for any connected planar
subdivision, there exists a data structure of size $O(n)$ that can
answer point location queries using $\tilde H + O(\tilde H^{2/3}+1)$
point/line comparisons.  Here, $\tilde H=\tilde H(G,D)$ is a lower
bound on the expected cost of any linear decision tree that solves
this problem.  Note that $\tilde H$ is often greater than the quantity
$H$ defined above and this is necessarily so.  To see this, consider
that the problem of testing whether a query point is contained in a
simple polygon $P$ with $n$ vertices is a special case of planar point
location in a connected planar subdivision.  However, in this special
case the subdivision only has 2 faces, so $H\le 1$.  It seems unlikely
that, for any simple polygon $P$ and any probability measure $D$ over $\R^2$,
it is always possible to test in $O(1)$ expected time if a point
$p$ drawn from $D$ is contained in $P$.  Indeed, it is not hard to
design a convex polygon $P$ and distribution $D$ so that the expected
cost of any algebraic decision tree for point location in $P$, under
query distribution $D$, is $\Omega(\log n)$.

Note that all known algorithms for planar point location that do not
place special restrictions on the input subdivision can be described
in the linear decision tree model of computation.\footnote{Although
significant breakthroughs have recently been made in this area
\cite{c06,p06}, we deliberately do not survey algorithms that require
the vertices of the subdivision to be on integer coordinates.}  The
data structures presented in the current paper are the most general
results known about planar point location and imply, to within a
lower order term, all of the results discussed in the introduction.

We achieve our results by showing how to compute a Steiner
triangulation $\Delta=\Delta(G,D)$ of $G$ that has nearly minimum
entropy over all possible triangulations of $G$ and then proving that
the entropy of a minimum-entropy Steiner triangulation of $G$ is a
lower bound on the cost of any linear decision tree for point location
in $G$.  By then applying the recent result of Arya \etal\ to the
Steiner triangulation $\Delta$ we obtain upper and lower bounds that
match to within a lower-order term.

A preliminary version of this paper, which dealt only with convex
subdivisions, has appeared in the Proceedings of the 19th ACM-SIAM
Symposium on Discrete Algorithms (SODA~2008) \cite{cdilm08}.

The remainder of this paper is organized as follows:  \Secref{prelim}
presents definitions and notations used throughout the paper.
\Secref{polygons} shows how to compute a near-minimum-entropy
triangulation of a simple polygon.  Finally, \secref{subdivisions}
applies these tools to obtain our point location structure for
connected planar subdivisions.

%%%%%%%%%%%%%%%%%%%%%%%%%%%%%%%%%%%%%%%%%%%%%%%%%%%%%%%%%%%%%%%%%%%
\section{Preliminaries}
\seclabel{prelim}

In this section we give definitions, notation, and background
required in subsequent sections.

\paragraph{Interiors and Boundaries.}
For a set $P\subseteq \R^2$, we denote the boundary of $P$ by $\boundary P$
and the interior of $P$ by $\interior(P)$.  The closure of $P$
is denoted by $\closure(P) = P\cup{\boundary P}$.

\paragraph{Triangles and Convex Polygons.}  For the purposes of this
paper, a \emph{triangle} is the common intersection of at most 3
closed halfplanes.  This includes triangles with infinite area and
triangles having 0, 1, 2, or 3, vertices. Similarly, a \emph{convex
$k$-gon} is the common intersection of at most $k$ closed halfplanes.

For a closed region $X\subseteq \R^2$, a \emph{triangulation} of $X$
is a set of triangles whose interiors are pairwise disjoint and whose
union is $X$.  We use the convention that, unless $X$ is explicitly
mentioned, the triangulation in question is a triangulation of $\R^2$.
This definition of a triangulation is often referred to as a
\emph{Steiner triangulation} since it allows vertices of the triangles
to be anywhere in $X$, and not at some finite predefined set of
locations.

\paragraph{Simple Polygons, Pseudotriangles, and Geodesic Triangles.}

A \emph{(near-simple) polygon} $P$ is a closed subset of $\R^2$ whose
boundary is piecewise linear and such that $\interior(P)$ is
homeomorphic to an open disk.  Note that this definition of a polygon
implies that every bounded face of a connected planar subdivision is a
polygon.  Also, triangles, as defined above, are polygons.  Note that
near-simple polygons are slightly more general than simple polygons,
for which $\boundary P$ is a simple closed curve.  However, our
definition is sufficiently close that algorithms designed for simple
polygons continue to work with near-simple polygons.

A reflex chain in a polygon $P$ is a consecutive sequence of vertices
$p_i,\ldots,p_j$ of $P$, where the internal angle at $p_k$ is at least
$\pi$, for all $k\in\{i+1,\ldots,j-1\}$. A \emph{pseudotriangle} is a
polygon whose boundary consists of 3 reflex chains.  An $i$-convex
pseudotriangle ($i\in\{0,1,2,3\}$) is a pseudotriangle in which $i$ of
the reflex chains consist of single line segments.

A \emph{shortest path} between points $a,b\in P$, denoted
$\overline{ab}_P$ is a curve of minimum length that is contained in
$P$ and that has $a$ and $b$ as endpoints.  For 3 points, $a,b,c\in
P$, a \emph{geodesic triangle} in $P$, denoted $\triangle_P abc$ is
the union of all shortest paths of the form $\overline{xc}_P$, where
$x\in\overline{ab}_P$.  Geodesic triangles are closely related to
pseudotriangles.  In particular, every geodesic triangle $t$ consists
of a pseudotriangle $\z t$ and three paths joining the three convex
vertices of $\z t$ to $a$, $b$, and $c$.

\paragraph{Classification Problems and Classification Trees.}

A \emph{classification problem} over a domain $\mathcal{D}$ is a
function $\mathcal{P}:\mathcal{D}\mapsto \{0,\ldots,k-1\}$.  A
$d$-ary \emph{classification tree} is a full $d$-ary tree\footnote{A
full $d$-ary tree is a rooted ordered tree in which each non-leaf node
has exactly $d$ children.} in which each internal node $v$ is labelled
with a function $P_v:\mathcal{D}\mapsto\{0,.\ldots,d-1\}$ and for
which each leaf $\ell$ is labelled with a value
$d(\ell)\in\{0,\ldots,k-1\}$. The \emph{search path} of an input $p$
in a classification tree $T$ starts at the root of $T$ and, at each
internal node $v$, evaluates $i=P_v(p)$ and proceeds to the $i$th
child of $v$.  We denote by $T(p)$ the label of the final (leaf) node
in the search path for $p$.  We say that the classification tree $T$
\emph{solves} the classification problem $\mathcal{P}$ over the domain
$\mathcal{D}$ if, for every $p\in \mathcal{D}$, $\mathcal{P}(p)=T(p)$.

In this paper, we are especially concerned with \emph{linear decision
trees}. These are binary classification trees for a problem
$\mathcal{P}$ over the domain $\R^2$.  Each internal node $v$ of a
linear decision tree contains a linear inequality $P_v(x,y)=ax+by \ge
c$, and the node evaluates to 1 or 0 depending on whether the query
point $(x,y)$ satisfies the inequality or not, respectively.
Geometrically, each internal node of $T$ is labelled with a directed
line and the decision to go to the left or right child depends on
whether $p$ is to the left or right (or on) this line.  An immediate
consequence of this is that, for each leaf $\ell$ of $T$, the closure
of $r(\ell)$ is a convex polygon. 

\paragraph{Probability.}

Throughout this paper $D$ is a probability measure over $\R^2$ that
represents the query distribution.  The notation $\Pr(X)$ denotes the
probability of event $X$ under the probability measure $D$.  The
notation $\Pr(Y|X)$ denotes the conditional probability of $Y$ given
$X$, i.e., $\Pr(Y|X)=\Pr(Y\cap X)/\Pr(X)$.  For any set $S$, we use
the shorthand $\cup S$ to denote $\bigcup_{s\in S} s$.

For a set $S$ of subsets of $\R^2$, we define the \emph{induced
entropy of $S$}, denoted by $H(S)$ as $H(S)=\sum_{s\in S}\Pr(s|{\cup
S})\log(1/\Pr(s|{\cup S}))$.  For two sets $S_1,S_2\subseteq\R^2$ with
$\cup S_1=\cup S_2$, the \emph{joint entropy} of $S_1$ and $S_2$, is
$H(S_1,S_2) = \sum_{s_1\in S_i}\sum_{s_2\in S_2} \Pr(s_1\cap
s_2)\log(1/\Pr(s_1\cap s_2))$.  It is well-known that $H(S_1,S_2)\le
H(S_1)+H(S_2)$ (see, for example, Gray \cite[Lemma~2.3.2]{g08}).

We will sometimes abuse terminology slightly by referring to a
triangulation $\Delta$ of $X$ as a \emph{partition} of $X$ into
triangles, although strictly speaking this is not true since the
triangles in $\Delta$ are closed sets that overlap at their
boundaries.  We will then continue the abuse by computing the induced 
entropy of $\Delta$.
This introduces a technical difficulty in that
$\sum_{t\in\Delta}\Pr(t)\ge 1$ and inequality is possible if there
exists sets $Y\subset \R^2$ such that the area of $Y$ is 0 and
$\Pr(Y)>0$.  To avoid this technical difficulty, we will assume that
$D$ is \emph{nice} in the sense that, if the area of $Y$ is 0 then
$\Pr(Y)=0$.   This implies that, for every $t$ in $\Delta$ $\Pr(t) =
\Pr(\interior(t))$.  This assumption will avoid lengthy technical but
uninteresting cases in our analysis.  In practice, this problem can
be avoided by using a symbolic perturbation of the query point.

The probability measures used in this paper are usually defined over
$\R^2$.  We make no assumptions about how these measures are
represented, but we assume that an algorithm can, in constant time,
perform each of the following two operations:
\begin{enumerate}
\item given a triangle $t$, compute $\Pr(t)$, and
\item given a triangle $t$ and a point $x$ at the
intersection of two of $t$'s supporting lines, compute a line $\ell$
that contains $x$ and that partitions $t$ into two open triangles
$t_0$ and $t_1$ such that $\Pr(t_0)\le\Pr(t_1)\le\Pr(t)/2$.
\end{enumerate}
Requirement 2 is used only for convenience in describing our data
structure.  It is not strictly necessary, but its use greatly
simplifies the exposition of our results.  To eliminate requirement 2,
one can use the same method described by Collette \etal\
\cite[Section~5]{cdilm08}.

For a classification tree $T$ that solves a problem
$P:\mathcal{D}\mapsto\{0,\ldots,k-1\}$ and a probability measure $D$
over $\mathcal{D}$, the \emph{expected search time} of $T$ is the
expected length of the search path for $p$ when $p$ is drawn at random
from $\mathcal{D}$ according to $D$.  Note that, for each leaf $\ell$
of $T$ there is a maximal subset $r(\ell)\subseteq \mathcal{D}$ such
that the search path for any $p\in r(\ell)$ ends at $\ell$.  Thus, the
expected search time of $T$ (under distribution $D$) can be written as
\[
     \mu_D(T) = \sum_{\ell\in L(T)} \Pr(r(\ell))\times \depth(\ell)
	\enspace ,
\]
where $L(T)$ denotes the leaves of $T$ and $\depth(\ell)$ denotes the
length of the path from the root of $T$ to $\ell$.  For any tree $T$
we use $V(T)$ to denote the vertices of $T$.

The following theorem, which is a restatement of (half of) Shannon's
Fundamental Theorem for a Noiseless Channel \cite[Theorem 9]{s48}, is
what all previous results on distribution-sensitive planar point
location use to establish their optimality:

\begin{thm}[Fundamental Theorem for a Noiseless Channel]\thmlabel{shannon}
Let $\mathcal{P}:\mathcal{D}\mapsto \{0,\ldots,k-1\}$ be a classification
problem and let $p\in \mathcal{D}$ be selected from a distribution $D$ such
that $\Pr\{\mathcal{P}(p)= i\}=p_i$, for $0\le i< k$.  Then, any
$d$-ary classification tree $T$ that solves $\mathcal{P}$ has
\begin{equation}
     \mu_D(T) \ge \sum_{i=0}^{k-1} p_i\log_d(1/p_i) \enspace .
	\eqlabelx{shannon}
\end{equation}
\end{thm}
\thmref{shannon} is typically applied to the point location problem by
treating point location as the problem of classifying the query point
$p$ based on which face of $G$ contains it.  In this way, we obtain
the lower bound in \eqrefx{entropy-face}.

%%%%%%%%%%%%%%%%%%%%%%%%%%%%%%%%%%%%%%%%%%%%%%%%%%%%%%%%%%%%%%%%%%%
\section{Minimum Entropy Triangulations} 
\seclabel{polygons}

Let $P$ be a simple polygon with $n$ vertices, denoted
$p_0,\ldots,p_{n-1}$ as they occur, in counterclockwise order, on the
boundary of $P$.  We will show how to find a triangulation of $P$ that
has near-minimum entropy.  That is, we will find a triangulation
$\Delta=\Delta(P,D)$ such that $H(\Delta)$ is near-minimum over all
triangulations of $P$.  In order to shorten the formulas in this
section, we will implicitly condition the distribution $D$ on $P$.
More precisely, throughout this section the notation $\Pr(X)$ should
be treated as shorthand for $\Pr(X|P)$.

\subsection{The Triangulation $\Delta=\Delta(P,D)$}
\seclabel{delta}

Our triangulation algorithm is recursive and takes as input a polygon
$P$ and a reflex chain $p_i,\ldots,p_j$ on the boundary of $P$.  If
$P$ is a triangle, then there is nothing to do, so the algorithm
outputs $P$ and terminates. Otherwise, the
algorithm first selects a point $p_k$ on the boundary of $P$ and adds
all the edges of the geodesic triangle $t=\triangle p_ip_jp_k$ to the
triangulation $\Delta$.
Observe that removing $t$ from $P$ disconnects $P$ into
components $P_1,\ldots,P_m$ where $\closure(P_i)$ is a polygon that shares a
reflex chain $C_i$ with the pseudotriangle $t$ (see
\figref{ii}).  The point
$p_k$ is selected in such a way that, for all $i\in\{1,\ldots,m\}$, $\Pr(P_i)
\le (1/2)\Pr(P)$.\footnote{The existence of such a point $p_k$ is
readily established by a standard continuity argument; see Bose \etal\
\cite{bdhlim07} for an example.} Each of the sub-polygons $P_1,\ldots,P_m$
can then be triangulated recursively by applying the algorithm to
$P_i$ and the reflex chain $C_i$.

\begin{figure}
  \begin{center}
    \begin{tabular}{cc}
      \includegraphics{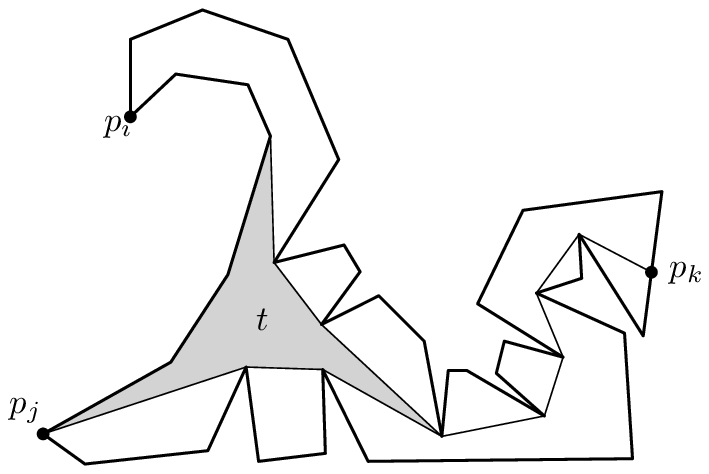} & \includegraphics{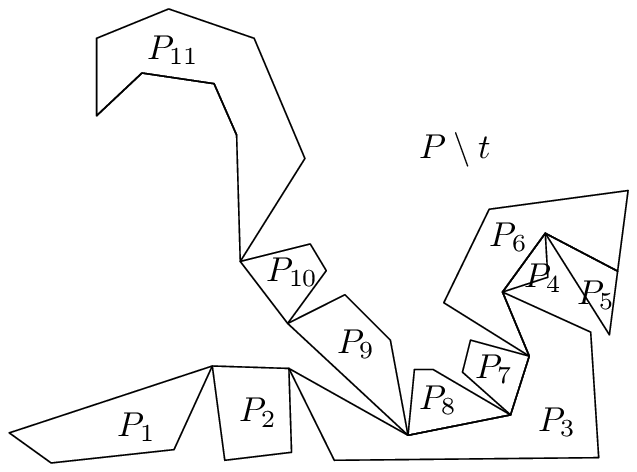} 
    \end{tabular}
  \end{center}
  \caption{The geodesic triangle $t=\triangle p_i p_j p_k$ partitions $P$ into several
pieces $P_1,\ldots,P_m$.}
  \figlabel{ii}
\end{figure}

To complete the triangulation $\Delta$ all that remains is to
partition $\z t=\closure(t\setminus \boundary t)$ into triangles.  To
do this, we first partition $\z t$ into at most one triangle $t'$ and
three
2-convex pseudotriangles $t_0,t_1,t_2$ as shown in
\figref{pt-partition}.a. Let $Q_i$ be the connected component of
$(\interior(P)\setminus \z t)\cup t_i$ that contains $t_i$.  To
complete the triangulation we will partition $t_i$ into triangles, for
each $i\in\{0,1,2\}$, using a recursive algorithm.  This algorithm
selects an edge $e_i$ of the reflex chain in $t_i$ and extends $e_i$
in both directions until it reaches the boundary of $t_i$ (see
\figref{pt-partition}.b).  The resulting line segment partitions $t_i$
into a triangle $t_i'$, and two 2-convex pseudotriangles $t_{i,0}$ and
$t_{i,1}$ that are triangulated recursively.  At the same time, $Q_i$
is partitioned into up to 4 pieces (see \figref{pt-partition}.c):

\begin{enumerate}
\item the triangle $t_i'$, and
\item a subpolygon $P_j$ incident to $e_i$,
\item The two connected components $Q_{i,0}$
and $Q_{i,1}$ of $Q_i\setminus t_i'$ that contain $t_{i,0}$ and
$t_{i,1}$, respectively.
\end{enumerate}
The edge $e_i$ is selected
so that $\Pr(Q_{i,b})\le (1/2)\Pr(Q_i)$ for each
$b\in\{0,1\}$.\footnote{The existence of such an edge $e_i$ is assured
by yet another continuity argument.}  This
completes the description of the triangulation $\Delta$.  A partially
completed triangulation is show in \figref{delta-example}.

\begin{figure}
  \begin{center}
    \begin{tabular}{ccc}
      \includegraphics{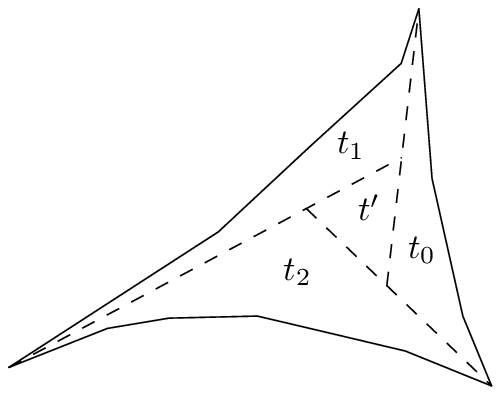} & 
      \includegraphics{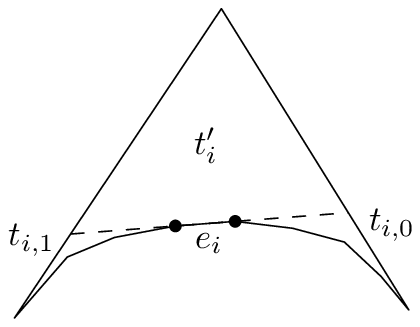} &
      \includegraphics{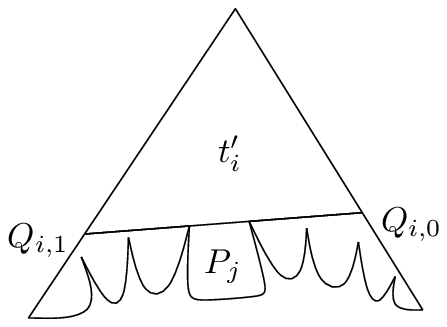} \\
      (a) & (b) & (c)
    \end{tabular}
  \end{center}
  \caption{Partitioning (a) a pseudotriangle $\z t$ into three 2-convex
pseudotriangles $t_0,t_1,t_2$ and one triangle $t'$ (b) a 2-convex
pseudotriangle $t_i$ into one triangle $t_i'$ and two 2-convex
pseudotriangles $t_{i,0}$ and $t_{i,1}$, and (c) $Q_i$ into 4 pieces.}
  \figlabel{pt-partition}
\end{figure}

\begin{figure}
  \begin{center}
      \includegraphics{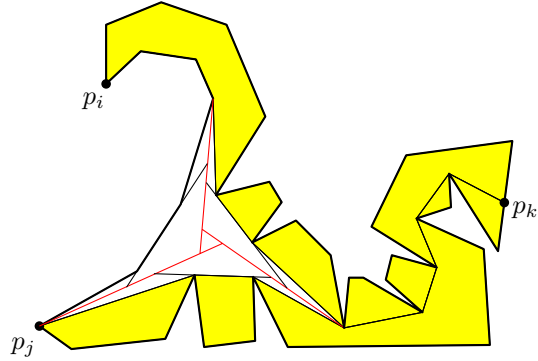}
  \end{center}
  \caption{The triangles obtained during the first level of recursive
triangulation.  The yellow subpolygons are triangulated recursively.}
  \figlabel{delta-example}
\end{figure}

\subsection{The $\Delta$-Tree $T=T(P,D)$}

In order to study the entropy of the triangulation $\Delta$ defined
above, we will impose a tree structure on the pieces of $P$ induced by
the triangles in $\Delta$.  The $\Delta$-tree $T=T(P,D)$ for $P$ is a
tree whose nodes are subpolygons of $P$ and which has the property
that, for any node $y$ that is the child of a node $x$, $y\subseteq
x$.

The tree $T$ has three different kinds of nodes, called P-nodes,
T-nodes, and Q-nodes.  The root $r$ of $T$ is the polygon $P$ and is
a \emph{P-node}.  The root of $T$ has the following children
(defined in terms of the construction algorithm in the previous
section; see \figref{delta-tree}):

\begin{enumerate}
\item Each subpolygon $P_i$ whose boundary does not share a segment
      with $\z t$ is a child of $r$ and is a P-node.  
\item The subpolygon $Q=\z t\cup Q_0\cup Q_1\cup Q_2$ is a child of $r$ and is
called a \emph{T-node}.
\end{enumerate}

The subpolygon $Q$ has three children $Q_0,Q_1,Q_2$ that are called
\emph{Q-nodes}.  The subtree rooted at $Q_i$ is a ternary tree
corresponding to the recursive partitioning of $t_i$ and $Q_i$ done by
the algorithm. The leaves of this subtree are P-nodes and the internal
nodes of this subtree are Q-nodes.  Each internal node has up to 3
children, up to 1 of which may be a P-node corresponding to a
subpolygon $P_j$ and up to two of which may be Q-nodes.

Note that the above definition yields a tree whose leaves are P-nodes
that correspond to the subpolygons $P_1,\ldots,P_m$  obtained by
removing $t$ from $P$.  The subtree rooted at each such leaf is
obtained recursively from the recursive triangulation of $P_i$.

\begin{figure}
 \begin{center}\includegraphics{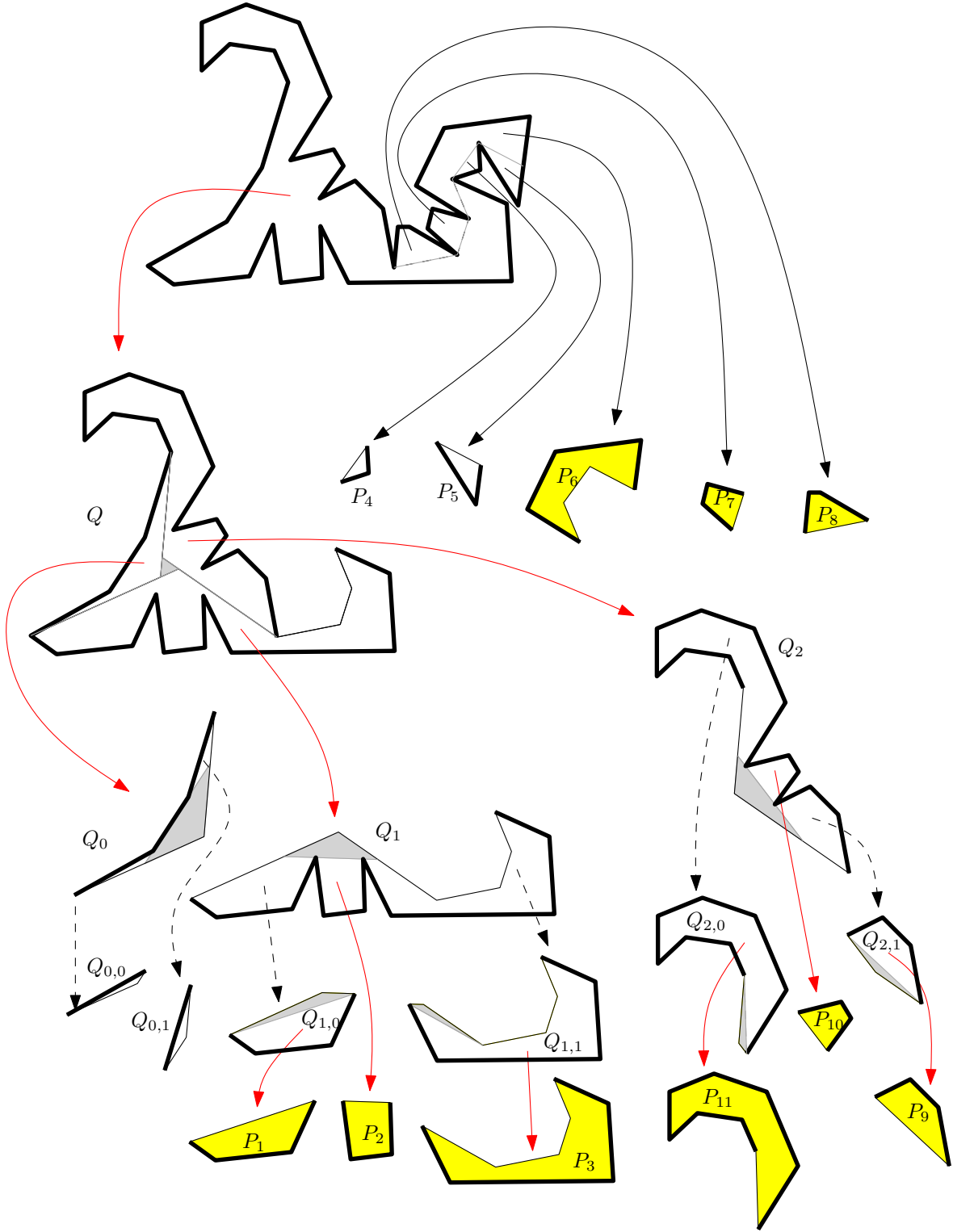}\end{center}
   \caption{The $\Delta$-tree $T$.  Yellow leaves in this tree are the
     root of subtrees obtained recursively.  Grey areas show the
portions of a node not covered by its children. 
A black edge from a node $x$ to a node $y$ indicates that $\Pr(y)\le
(1/2)\Pr(x)$.  A solid edge leading into a node $x$ indicates that $x$
is separated from the rest of $P$ by a shortest path in $P$.}
  \figlabel{delta-tree}
\end{figure}

Now that we have defined the tree $T$, we study some of its
properties.  Our first lemma says that $T$ does a good job of
splitting $P$ based on its probabilities.

\begin{lem}\lemlabel{t-halving}\lemlabel{A}
Let $P$ be a polygon, let $D$ be a probability measure over $\R^2$, 
and let $T=T(P,D)$ be the $\Delta$-tree for $(P,D)$.
Let $x$ be a node of $T$ whose depth is $i$. Then $\Pr(x) \le
(1/2^{\floor{i/4}})\Pr(P)$.
\end{lem}

\begin{proof}
Assume $i\ge 4$, otherwise there is nothing to prove. Let $r'$ be the
fourth node on the path from the root, $r$, of $T$ to $x$.  Let $P_{r'}$ be
the path in $T$ from $r$ to $r'$.  If $P_{r'}$ has at least two
P-nodes, then $\Pr(r')\le (1/2)\Pr(r)=(1/2)\Pr(P)$. Otherwise, the
second node in $P_{r'}$ is a T-node followed by 2 Q-nodes.  By
construction, for any Q-node $y$ whose parent is a Q-node $z$,
$\Pr(y)\le(1/2)\Pr(z)$.  Therefore, $\Pr(r')\le (1/2)\Pr(\prnt(r'))\le
(1/2)\Pr(r) = (1/2)\Pr(P)$.  If $x=r'$, then $i=4$ and the proof is complete.
Otherwise, apply the same argument inductively on the path from $r'$ to $x$, to
obtain
\[
   \Pr(x) \le \Pr(r') \cdot 1/2^{\floor{(i-4)/4}} 
    =  \Pr(r')\cdot 1/2^{\floor{i/4}-1}  
   \le (1/2^{\floor{i/4}})\Pr(r) 
    = (1/2^{\floor{i/4}})\Pr(P)  \enspace ,
\]
and this completes the proof.
%Let $P_x$ be the path from the root, $r$, of $T$ to $x$.  Suppose, for
%now, that $P_x$ contains at least 2 P-nodes (one of which is the root
%$r$).  Consider the path $P'$ from $r$ to the second P-node $r'$ in
%$P_x$.  By construction, $\Pr(r') \le (1/2)\Pr(r)$, so if $P'$ is of
%length at most 4 then we can (inductively) apply the same argument to
%the path from $r'$ to $x$, so that 
%\[
%   \Pr(x) \le \Pr(r') \cdot 1/2^{\floor{(i-4)/4}} 
%    =  \Pr(r')\cdot 1/2^{\floor{i/4}-1}  
%   \le (1/2^{\floor{i/4}})\Pr(r) .
%    = (1/2^{\floor{i/4}})\Pr(P) .
%\]
%and we are done.
%
%Otherwise, the second node on the path $P'$ is a T-node, which is followed by
%$k\ge 1$ Q-nodes, followed finally by $r'$.  Note that, for any Q-node
%$x$ whose parent is a Q-node $y$, $\Pr(x)\le (1/2)\Pr(y)$.  Therefore,
%$\Pr(r') \le (1/2^{k-1})$.  Since the length of the $P'$ is $k+2 \ge 4$,
%and $k-1 \ge (k+2)/4$ for all $k\ge 2$,
%we get
%\[
%    \Pr(r') \le (1/2^{k-1})\Pr(r) \le (1/2^{(k+2)/4})\Pr(r) \enspace ,
%\]
%and we can then argue inductively on the path from $r'$ to $x$ (whose length
%is $i-k-2$).
%
%Finally, if $P_x$ contains only a single P-node then, arguing as
%above, we see that there is a sequence of at least $i-1$ consecutive
%$Q$ nodes in $P_x$, so 
%\[
%   \Pr(x)\le 1/2^{i-2}\Pr(r)\le 1/2^{i/4}\Pr(r) \le 1/2^{\floor{i/4}}\Pr(r)
%   = 1/2^{\floor{i/4}}\Pr(P)
%\]
%for all $i\ge 4$.
\end{proof}

Our next lemma says that a single line segment does not intersect
very many high probability triangles in $\Delta$.

\begin{lem}\lemlabel{t-segment-x}
Let $P$ be a polygon, let $D$ be a probability measure over $\R^2$, 
and let $T=T(P,D)$ be the $\Delta$-tree for $(P,D)$, let $s\subseteq P$ be
a line segment, and let $S_i\subseteq V(T)$ be the set of all vertices
$x\in V(T)$ that are distance at most $i$ from the root of $T$ and
such that
$\interior(x)$ intersect $s$.  Then $|S_i|\le 6i^2$
\end{lem}

\begin{proof}
There are 3 types of nodes in $T$ whose interiors intersect $s$:
Type~1~nodes contain one endpoint of $s$ in their interior, Type~2
nodes contain both endpoints of $s$ in their interior, and Type~0
nodes contain no endpoints of $s$ in their interior.  Notice that each
level of $T$ contains at most 1 Type~2 node or 2 Type~1 nodes, so the
total number of Type~1 and Type~2 nodes at distance at most $i$ from
the root is at most $2i+1$.  Thus, all
that remains is to bound the number of Type~0 nodes whose distance
from the root is at most $i$.

Let $w$ be any P-node such that $\interior(w)$ does not contain either
endpoint of $s$.  Since $w$ is a P-node, there is a reflex chain $C_w$
on the boundary of $w$ that is a shortest path between two points on
the boundary of $P$, and every path in $P$ from $\interior(w)$ to $s$
intersects $C_w$. Stated another way, the interior of a P-node does
not intersect $s$ unless it contains at least one endpoint of $s$.
Therefore, all Type~0 nodes are either T-nodes or Q-nodes.

For every Type~0 node $x$, there is a path $P_x$ in $T$ from $x$ to a
T-node that is adjacent to a Type~1 or Type~2 node.  Furthermore, the
path $P_x$ consists of $x$ followed by 0 or more Q-nodes, and
terminates with a T-node.  Looking more closely at the definition of
Q-nodes, we see that two sibling Q-nodes $x$ and $y$ are not mutually
visible, i.e., there is no line segment $s\subseteq P$ that intersects
both $\interior(x)$ and $\interior(y)$.

All of this implies that each of the at most $2i+1$ Type~1 or Type~2 
nodes is adjacent to at most 1 Type~0 T-node, and this T-node is the
endpoint of at most 3 paths of Type~0 Q-nodes. Each such path is of
length at most $i$.  Therefore, the total number of nodes in $T$ that
intersect $s$ is at most  $2i\cdot 3i=6i^2$.
\end{proof}

\subsection{Minimum-Entropy Triangulation}

Next, we show that the triangulation $\Delta$ defined above is
nearly-minimum entropy over all possible triangulations of $P$.  We do
this by developing a technique for lower-bounding the entropy of one
triangulation in terms of the entropy of another triangulation.  We
then show how to apply this technique to lower bound the entropy of
any triangulation $\Delta^*$ in terms of the entropy of $\Delta$.
  
To obtain lower bounds on the entropy of a triangulation $\Delta^*$,
consider the following easily proven observation: If each triangle in
$\Delta^*$ intersects at most $c$ triangles of some triangulation
$\Delta$ then $H(\Delta^*) \ge H(\Delta) - \log c$.\footnote{Proof:
Consider the set $X=\{ t^*\cap t : t^*\in\Delta^*, t\in \Delta\}$.
Each triangle of $\Delta^*$ contributes at most $c$ pieces to $X$, so
we have $H(\Delta) \le H(X) \le H(\Delta^*) + \log c$.}  This
observation allows us to use $\Delta$ to prove a lower bound on the
entropy of a triangulation $\Delta^*$.  Unfortunately, the condition
that each triangle of $\Delta^*$ intersect at most $c$ triangles of
$\Delta$ is too restrictive for our purposes.  Instead, we require
following stronger result:

\begin{lem}\lemlabel{pieces}
Let $D$ be a probability measure over $\R^2$.  Let $\Delta$ and
$\Delta^*$ be triangulations, and let $\{\Delta_1,\ldots,\Delta_m\}$
be a partition of the triangles in $\Delta$.  Suppose that, for all
$i\in\{1,\ldots,m\}$ and for each triangle $t^*\in\Delta^*$, $t^*$
intersects at most $c_i$ triangles in $\Delta_i$.  Then
\begin{eqnarray*}
   H(\Delta) \le 
	 H(\Delta^*) + H(\{\cup\Delta_1,\ldots,\cup\Delta_m\}) + 
	\sum_{i=1}^m \Pr(\cup\Delta_i)\log c_i
 \enspace . 
\end{eqnarray*}
\end{lem}

Intuitively, \lemref{pieces} can be thought of as follows:  If we tell
an observer which of the $\Delta_i$ a point $p$ drawn according to $D$
occurs in then the amount of information we are giving the observer
about the experiment is at most
$H(\{\cup\Delta_1,\ldots,\cup\Delta_m\})$.  However, after giving away
this information, we are able to apply the simple observation in the
previous paragraph, since each triangle in $\Delta^*$ intersects at
most $c_i$ elements of each $\Delta_i$.  Thus, \lemref{pieces} is
really just $m$ applications of the simple observation.  The following
proof formalizes this:

\begin{proof}
\begin{eqnarray*}
   H(\Delta^*)  +H(\{\cup\Delta_1,\ldots,\cup\Delta_m\})
     & \ge & H(\Delta^*,\{\cup\Delta_1,\ldots,\cup\Delta_m\}) \\
     &  =  & \sum_{i=1}^m
              \sum_{t^*\in\Delta^*}
               \Pr(t^*\cap\Delta_i)\log (1/\Pr(t^*\cap\Delta_i) \\
     & \ge & \sum_{i=1}^m
              \sum_{t\in\Delta_i}
               \sum_{t^*\in\Delta^*}
                \left(
                  \Pr(t^*\cap t)\log (1/\Pr(t^*\cap t) - \log c_i
                \right) \\
     & \ge  & \sum_{i=1}^m 
               \sum_{t\in\Delta_i}
                \Pr(t)\log(1/\Pr(t)) 
                -\sum_{i=1}^m \Pr(\cup\Delta_i)\log c_i  \\
     &  =  & H(\Delta) -\sum_{i=1}^m \Pr(\cup\Delta_i)\log c_i
            \enspace ,
\end{eqnarray*}
and this completes the proof.
\end{proof}

The remainder of our argument involves partitioning the triangles of
$\Delta$ into subsets $\Delta_1,\ldots,\Delta_m$ and then showing that
$H(\Delta_1,\ldots,\Delta_m)$ and $\sum_{i=1}^m \Pr(\cup\Delta_i)\log
c_i$ are not too big.  To help us, we will use the $\Delta$-tree $T$.
For a node $x$ in $T$ with children $x_1,\ldots,x_k$, let $t(x) = x
\setminus (\bigcup_{i=1}^m x_i)$ be the portion of $x$ not covered by
$x$'s children.  Note that $t(x)$ is always either the empty set or is
a triangle in $\Delta$ (see \figref{delta-tree}).  In fact, for every
triangle $t\in\Delta$, there is exactly one $x\in V(T)$ such that
$t(x)=t$, and for every $x\in V(T)$ such that $t(x)$ is non-empty
there is exactly one $t\in\Delta$ such that $t(x)=t$.  This implies
that\footnote{Here, and throughout the remainder, we slightly abuse
notation by using the convention that $0\cdot\log(1/0)=0$.}
\[
    H(\Delta) = \sum_{t\in\Delta}\Pr(t)\log(1/\Pr(t)) =
       \sum_{x\in V(T)}\Pr(t(x))\log(1/\Pr(t(x))) \enspace .
\]
For a node $x\in V(T)$, we define $\Pr(x)=\Pr(t(x))$ is the
probability that a point drawn from $D$ is contained in $t(x)$.

Next we apply \lemref{pieces} to obtain a lower bound on the
entropy of any triangulation $\Delta^*$. 

\begin{lem}\lemlabel{min-H-triangulation}\lemlabel{Z}
Let $P$ be a simple polygon, let $D$ be a probability measure over
$\R^2$, and consider the triangulation $\Delta=\Delta(P,D)$.
Then, for any triangulation $\Delta^*$ of $P$,
\[
    H(\Delta) \le H(\Delta^*) + O(H(\Delta^*)^{2/3}+1) \enspace .
\]
\end{lem}

\begin{proof}
Let $T=T(P,D)$ be the $\Delta$-tree for $(P,D)$.
Partition the nodes of $T$ into
\emph{groups} $G_1,G_2,\ldots$ where
\[
	G_i = \{x\in V(T) : 1/2^{i} < \Pr(x) \le 1/2^{i-1} \} \enspace .
\]
In the following we will fix a value $\alpha$, $0 < \alpha < 1$, to be
defined later.  A group $G_i$ is \emph{large} if it contains at least
$2^{\alpha i}$ elements, otherwise $G_i$ is \emph{small}.  Let $I^+$
denote the index set of the large groups, i.e., $I^{+}=\{i\in\N :
|G_i| \ge 2^{\alpha i}\}$.  Let $I^{-}= \N\setminus I^+$ be the index
set of the small groups.

Note that, for any group $G_i$, \lemref{t-halving} ensures that all
elements of $G_i$ have depth at most $4i$ in $T$.  Therefore, 
\lemref{t-segment-x} ensures that any triangle of
$\Delta^*$ intersects at most $3\times 6\times (4i)^2=288i^2$ 
triangles of $G_{i}$.  Therefore,
applying \lemref{pieces} with $c_{i}=288i^2$, we obtain:
\begin{equation} 
 H(\Delta) \le 
   H(\Delta^*) + H(\{{\cup G_{i}} : i\in\N\}) 
   + \sum_{i=1}^{\infty}\Pr(\cup G_{i})\log(288i^2)  \enspace .
   \eqlabelx{doit}
\end{equation}
Thus, all that remains is to bound the contribution of the last two
terms on the right hand side of \eqrefx{doit}.  First,
\begin{eqnarray*}
   \sum_{i=1}^{\infty}\Pr(\cup G_{i})\log(288i^2)
   &   =  & \sum_{i=1}^\infty\sum_{t\in G_i}\Pr(t)\log(288i^2) \\
   &  \le  & \sum_{i=1}^\infty\sum_{t\in G_i}\Pr(t)(O(1)+\log\log(1/\Pr(t))) \\
    &  =  & \sum_{t\in\Delta} \Pr(t) (O(1)+\log\log(1/\Pr(t))) \\
    &  =  & O(1+\log H(\Delta)) \enspace ,
\end{eqnarray*}
where the last equality follows from Jensen's Inequality.
Finally, we show that the contribution of $\overline
H=H(\{\cup\{G_{i}\} : i\in\N\})$ is at most
$O(H(\Delta)^{2/3})$. 
\begin{eqnarray*}
\overline H 
 &  =  & H(\{\cup\{G_{i}\} : i\in\N\}) \\
 &  =  & \sum_{i=1}^\infty\Pr(\cup G_{i})\log(1/\Pr(\cup G_{i})) \\
 &  =  & \sum_{i\in I^+}\Pr(\cup G_{i})\log(1/\Pr(\cup G_{i})) 
         + \sum_{i\in I^-}\Pr(\cup G_{i})\log(1/\Pr(\cup G_{i})) \\
 & \le & \sum_{i\in I^+}\Pr(\cup G_{i})\log(1/\Pr(\cup G_{i})) 
         + \sum_{i\in I^-}2^{\alpha i}/2^{i-1}\log(2^{i}) \\
 & \le & \sum_{i\in I^+}\Pr(\cup G_{i})\log(1/\Pr(\cup G_{i})) 
         + \sum_{i=1}^\infty 2^{\alpha i}/2^{i-1}\log(2^{i}) \\
 &  =  & \sum_{i\in I^+}\Pr(\cup G_{i})\log(1/\Pr(\cup G_{i})) 
         + \sum_{i=1}^\infty i2^{\alpha i}/2^{i-1} \\
 &  =  & \sum_{i\in I^+}\Pr(\cup G_{i})\log(1/\Pr(\cup G_{i})) 
         + 2\cdot\sum_{i=1}^{\infty} i/2^{(1-\alpha)i} \\
 &  =  & \sum_{i\in I^+}\Pr(\cup G_{i})\log(1/\Pr(\cup G_{i})) 
         + 2\cdot\left(\frac{(1/2)^{1-\alpha}}{(1-(1/2)^{1-\alpha})^2}\right) \\
 & \le & \sum_{i\in I^+}\Pr(\cup G_{i})\log(1/\Pr(\cup G_{i})) 
         + O(1/(1-\alpha)^2) \enspace ,
\end{eqnarray*}
where the last equality is obtained using the Taylor series expansion
for $e^{x}$ to obtain the inequality $1-1/2^{x} \ge x\ln 2 - (x^2\ln 2)/2$
for $x$ close to 0.  Continuing, we get 
\begin{eqnarray*}
\overline H
  & \le & \sum_{i\in I^+}\Pr(\cup G_{i})\log(1/\Pr(\cup G_{i})) 
         + O(1/(1-\alpha)^2) \\
  & \le & \sum_{i\in I^+}\Pr(\cup G_{i})\log(2^i/|G_i|)
         + O(1/(1-\alpha)^2) \\
  & \le & \sum_{i\in I^+}\Pr(\cup G_{i})\log(2^i/2^{\alpha i})
         + O(1/(1-\alpha)^2) \\
  &  =  & \sum_{i\in I^+}\Pr(\cup G_{i})(1-\alpha) i
         + O(1/(1-\alpha)^2) \\
  &  =  & (1-\alpha)\sum_{i\in I^+}\Pr(\cup G_{i}) i
         + O(1/(1-\alpha)^2) \\
  &  =  & (1-\alpha)\sum_{i\in I^+}\Pr(\cup G_{i})\log(2^i)
         + O(1/(1-\alpha)^2) \\
  &  =  & (1-\alpha)\sum_{i\in I^+}\sum_{t\in G_i}\Pr(t)\log(2^i)
         + O(1/(1-\alpha)^2) \\
  &  =  & (1-\alpha)\sum_{i\in I^+}\sum_{t\in G_i}\Pr(t)\log(1/\Pr(t))
         + O(1+1/(1-\alpha)^2) \\
  & \le & (1-\alpha)H(\Delta) + O(1+1/(1-\alpha)^2) \\
  & \le &  O(H(\Delta)^{2/3}+ 1)
\end{eqnarray*} 
Where the last inequality is obtained by setting 
$\alpha=1-1/H(\Delta)^{1/3}$.  Thus, we have shown that
\begin{equation}
  H(\Delta) \le H(\Delta^*) + O(H(\Delta)^{2/3}+1) \enspace ,
   \eqlabelx{almost-last}
\end{equation}
which implies that $H(\Delta) = O(H(\Delta)^*+1)$.  Applying this to
the right hand side of \eqrefx{almost-last} yields $H(\Delta) \le
H(\Delta^*) + O(H(\Delta^*)^{2/3} + 1)$, completing the proof.
\end{proof}

\lemref{min-H-triangulation} shows that the triangulation
$\Delta=\Delta(P,D)$ defined previously is nearly minimum-entropy over
all triangulations of $P$.  The following theorem gives an algorithmic
version of \lemref{min-H-triangulation}.

\begin{thm}\thmlabel{min-H-triangulation}
Let $P$ be a simple polygon with $n$ vertices, and let $D$ be a
probability measure over $\R^2$.  Then there exists an $O(n\log n)$
time algorithm that computes a triangulation $\Delta'$ of $P$ having
$O(n)$ triangles and such that, for any triangulation $\Delta^*$ of
$P$,
\[
    H(\Delta') \le H(\Delta^*) + O(H(\Delta^*)^{2/3}+1) \enspace .
\]
\end{thm}

\begin{proof}
We show how the construction of the triangulation $\Delta$ described in
\secref{delta} can be modified to run in $O(n\log n)$ time.  When
constructing $\Delta$ the first step is to find the third vertex $p_k$ of the
geodesic triangle $t=\triangle p_i p_j p_k$.  This can be accomplished
in $O(n)$ time by computing the shortest path trees from $p_i$ and
$p_j$ to all other vertices of $P$ and using these to find $p_k$.  For
an example of a similar computation, see Bose \etal\
\cite[Section~2.2]{bdhlim07}.

Next, $\z t$ is split into three 2-convex pseudotriangles $t_0,t_1,t_2$,
which is easily accomplished in $O(n)$ time.  The last step, before
recursing, is to triangulate each of $t_0,t_1,t_2$.  This step can be
accomplished in $O(n)$ time using a 2-sided exponential searching
trick that was used by Mehlhorn \cite{m75} in the construction of
biased binary search trees (see also, Collette \etal\
\cite[Theorem~1]{cdilm08}).

Finally, the algorithm recurses on each of the pieces
$P_1,\ldots,P_m$.  In this way, we obtain a divide-and-conquer
algorithm for constructing $\Delta$.  Unfortunately, this algorithm
may have running time $\Omega(n^2)$ since there is no bound
significantly smaller than $n$ on the size of an individual subproblem
$P_i$.  To overcome this, before recursing on a subproblem $P_i$ we
check if it contains more than $n/2$ vertices.  If so, then rather
than recursing normally on $P_i$ we choose a geodesic triangle $t^*$,
one of whose sides is the reflex chain $C_i$ and such that removing
$t^*$ from $P_i$ leaves a set of subpolygons
$P_{i,1},\ldots,P_{i,m_i}$ each with at most $n/2$ vertices.  This
modification then yields an algorithm whose recursion tree has depth
$O(\log n)$ and at which the work done at each level is $O(n)$, so the
total running time of this algorithm is $O(n\log n)$.

Note that this algorithm yields a triangulation $\Delta'$ that is
different from $\Delta$.  In particular, there may exist one $P_{i,j}$
with $\Pr(P_{i,j})>\Pr(P_i)/2$.  Despite this, all the proofs of
Lemmas~\ref{lem:A}--\ref{lem:Z} continue to hold almost without modification.
The only difference occurs in \lemref{t-halving}, which now only
guarantees a bound of $2^\floor{i/8}$ on the number of black edges,
but this has almost no effect on subsequent computations.

Finally, to see that $\Delta'$ contains $O(n)$
triangles, we count the different types of edges used in the
triangulation $\Delta'$.  Some of these edges are edges of $P$, of
which there are at most $n$.  Some of these edges are edges of
geodesic triangles,
which always connect two vertices of $P$ and do not cross each other,
so there are at most $n-2$ of these.  The remaining edges are used to
triangulate the interiors of pseudotriangles.  A pseudotriangle that has $k$ vertices
is triangulated using $3 + 2(k-3)$ edges.  Since the total number of
vertices in all pseudotriangles is at most $2n$, this means that there are at
most $6n$ edges used to triangulate pseudotriangles.  Therefore, the total
number of edges used by triangles in $\Delta'$, and hence the number
of triangles in $\Delta'$, is $O(n)$.  
\end{proof}

%%%%%%%%%%%%%%%%%%%%%%%%%%%%%%%%%%%%%%%%%%%%%%%%%%%%%%%%%%%%%%%%%%%
\section{Point Location in Simple Planar Subdivisions}
\seclabel{subdivisions}

Next we consider the problem of point location in simple
subdivisions.  The following theorem of Arya~\etal~\cite{ammw07}
shows that a low entropy triangulation can be used to make a good
point location structure.

\begin{thm}[Arya \etal\ 2007]\thmlabel{ammw07}
Let $D$ be a probability measure over $\R^2$ and let $\Delta$ be a
triangulation of $\R^2$ having a total of $n$ triangles.  Then there exists a
data structure of size $O(n)$ that can be constructed in $O(n\log n)$
time, and for which the expected number of point/line comparisons
required to locate the face of $G$ containing a query point $p$, drawn
according to $D$, is $H(\Delta) + O(H(\Delta)^{1/2}+1)$.
\end{thm}

The following lemma shows that the entropy of a minimum-entropy
triangulation gives a lower bound on the cost of any point location
structure.

\begin{lem}\lemlabel{triangulate}
Let $T^*$ be any linear decision tree for a classification problem
$\mathcal{P}$ over $\R^2$.  Then there exists a linear decision tree
$T'$ for $\mathcal{P}$, such that, for each leaf $\ell$ of $T'$,
$\closure(r(\ell))$ is a triangle and $T'$ satisifies
\[
    \mu_D(T') \le \mu_D(T^*) + O(\log\mu_D(T^*))
\]
for any probability measure $D$ over $\R^2$.
\end{lem}

\begin{proof}
Each leaf $\ell$ of $T^*$ has a region $r(\ell)$ that is a convex
polygon.  If $r(\ell)$ has $k$ sides then the depth of $\ell$ in $T$
is at least $k$.  To obtain the tree $T'$ replace each such leaf
$\ell$ of $T^*$ by a balanced binary tree of depth $O(\log k)$ by
repeatedly splitting the leaf into two children $\ell_1$ and $\ell_2$
whose regions have $\ceil{(k+2)/2}$ and $\floor{(k+2)/2}$ vertices.
For a leaf $\ell\in L(T^*)$, let $s(\ell)$ denote the set of leaves in $T'$
in the subtree of $\ell$.   Then
\begin{eqnarray*}
   \mu_D(T^*) 
     &  =  & \sum_{\ell\in L(T^*)} \Pr(r(\ell))\cdot \depth(\ell) \\
     &  =  & \sum_{\ell\in L(T^*)}\sum_{\ell'\in s(\ell)} 
              \Pr(r(\ell'))\cdot \depth(\ell) \\
     & \ge & \sum_{\ell\in L(T^*)} 
             \sum_{\ell'\in s(\ell)}\Pr(r(\ell'))\cdot (\depth(\ell')
                   - O(\log (\depth(\ell)))) \\
     &  =  & \mu_D(T') - \sum_{\ell\in L(T^*)} 
             \sum_{\ell'\in s(\ell)}\Pr(r(\ell'))\cdot O(\log (\depth(\ell))) \\
     &  =  & \mu_D(T') - \sum_{\ell\in L(T^*)} 
             \Pr(r(\ell))\cdot O(\log (\depth(\ell))) \\
     & \ge & \mu_D(T') - O(\log(\mu_D(T^*)) \enspace , 
\end{eqnarray*}
where the last inequality is an application of Jensen's Inequality.
\end{proof}

\lemref{triangulate} says that for any linear decision tree for point
location, there is an underlying triangulation.  The entropy of this
triangulation gives a lower bound on the cost of the decision tree.
Thus, the entropy of a minimum entropy triangulation gives a lower
bound on the expected cost of any linear decision tree for point
location.

Keeping the above in mind, our point location structure is simple.
Let $G$ be a connected planar subdivision whose faces are
$F=\{F_1,\ldots,F_m\}$ and let $D$ be a probability measure over
$\R^2$.  We assume, without loss of generality that the outer face of
$G$ is the complement of a triangle, since otherwise we can add at
most 3 vertices and 4 edges to $G$ to make this true.  Adding these
edges will not increase the entropy the minimum weight triangulation
of $G$ by more than a constant.  With this assumption, testing if the
query point is in the outer face of $G$ can be done using 3 linear
comparisons after which we may safely assume that the query point is
contained in an internal face of $G$.

We triangulate each internal face $F_i$ of $G$ (a near-simple polygon)
using \thmref{min-H-triangulation} to obtain a triangulation
$\Delta_i$. The union of all $\Delta_i$
is a triangulation $\Delta$ of $\R^2$, to which we apply
\thmref{ammw07} to obtain a point location structure $R=R(G,D)$ for
point location in $\Delta$ and hence also in $G$.  The following
theorem shows that $R$ is nearly optimal:

\begin{thm}
Given a connected planar subdivision $G$ with $n$ vertices and a probability
measure $D$ over $\R^2$, a data structure $R=R(G,D)$ of size $O(n)$ can be
constructed in $O(n\log n)$ time that answers point location queries in $G$.
The expected number of point/line comparisons performed by $R$, 
for a point $p$ drawn according to $D$ is 
\[
  \mu_D(R) \le \mu_D(T^*) + O(\mu_D(T^*)^{2/3}+1) \enspace , 
\] 
where $T^*$ is any linear classification tree that answers point
location queries in $G$.
\end{thm}

\begin{proof}
The space and preprocessing requirements follow from
\thmref{min-H-triangulation} and \thmref{ammw07}.
To prove the bound on the expected query time, apply
\lemref{triangulate} to the tree $T^*$ and consider the resulting tree
$T'$, each of whose leaves have regions that are triangles and such
that
\begin{equation}
     \mu_D(T') \le \mu_D(T^*) + O(\log \mu_D(T^*)) \enspace .
       \eqlabelx{t-triangle}
\end{equation}
Observe that each leaf of $T'$ corresponds to a triangle in $\R^2$
that is completely contained in one of the faces of $G$.  Let
$\Delta'$ denote this set of triangles and let $\Delta'_i$ denote the
subset of $\Delta'$ contained in $F_i$.
Consider the entropy $H(\Delta')$ of the distribution induced by the
leaves of $T'$:
\begin{equation}\eqlabelx{h-delta-prime}
 \begin{split}
   H(\Delta') 
     & =  \sum_{i=1}^m \sum_{t\in \Delta'_i}\Pr(t)\log(1/\Pr(t)) \\
     & =  \sum_{i=1}^m \Pr(F_i)\sum_{t\in \Delta'_i}
            \Pr(t|F_i)\log(1/\Pr(t)) \\
     & =  \sum_{i=1}^m \Pr(F_i)\sum_{t\in \Delta'_i}
            \Pr(t|F_i)
            \left(
              \log(1/\Pr(t|F_i))-\log(\Pr(F_i))
            \right) \\
     & =  \sum_{i=1}^m \Pr(F_i)\sum_{t\in \Delta'_i}
            \Pr(t|F_i)\log(1/\Pr(t|F_i)) 
             + \sum_{i=1}^m \Pr(F_i) \log(1/\Pr(F_i)) \\
     & =  \sum_{i=1}^m \Pr(F_i) H(\Delta'_i) + H(F) \enspace .
 \end{split}
\end{equation}
Similarly, the entropy of $\Delta$ is given by 
\begin{eqnarray*}
   H(\Delta) 
     & = & \sum_{i=1}^m \sum_{t\in \Delta_i}\Pr(t)\log(1/\Pr(t)) \\
     & = & \sum_{i=1}^m \Pr(F_i)\sum_{t\in \Delta_i}
            \Pr(t|F_i)\log(1/\Pr(t|F_i)) 
          + H(F) \\
     & = & \sum_{i=1}^m \Pr(F_i) H(\Delta_i) + H(F) \enspace .
      \eqlabelx{h-delta}
\end{eqnarray*}
By \thmref{min-H-triangulation}, the triangles in $\Delta_i$ form a
nearly-minimum entropy triangulation of $F_i$.  More specifically, 
\begin{equation}
   H(\Delta_i) \le H(\Delta'_i) + O(H(\Delta'_i)^{2/3}+1)  \enspace .
    \eqlabelx{h-delta-i}
\end{equation}
Putting this all together, we have
\[
 \begin{aligned}
  H(\Delta) 
    &  =  \sum_{i=1}^m\Pr(F_i) H(\Delta_i) + H(F) 
             \\ 
    & \le \sum_{i=1}^m\Pr(F_i) (H(\Delta'_i) 
             + O(H(\Delta'_i)^{2/3}+1)) + H(F) 
             && \text{\hfill{(by \eqrefx{h-delta-i})}} \\ 
    &  =  H(\Delta') + \sum_{i=1}^m\Pr(F_i) O(H(\Delta'_i)^{2/3}+1)
             && \text{(by \eqrefx{h-delta-prime})} \\ 
    &  =  H(\Delta') + 
             \left(
               \sum_{i=1}^m\Pr(F_i) O(H(\Delta'_i)
             \right)^{2/3} + O(1)
             && \text{(by Jensen's Inequality)} \\
    &  =  H(\Delta') + 
             \left(
               O(1) \cdot
                \sum_{i=1}^m\Pr(F_i)
                 \sum_{t'\in\Delta'_i}
                  \Pr(t'|F_i)\log(1/\Pr(t'|F_i))
             \right)^{2/3} + O(1)
              \\
    &  =  H(\Delta') + 
             \left(
               O(1) \cdot
                \sum_{i=1}^m
                 \sum_{t'\in\Delta'_i}
                  \Pr(t')\log(\Pr(F_i)/\Pr(t'))
             \right)^{2/3} + O(1)
              \\
    & \le H(\Delta') + 
             \left(
               O(1) \cdot
                \sum_{i=1}^m
                 \sum_{t'\in\Delta'_i}
                  \Pr(t')\log(1/\Pr(t'))
             \right)^{2/3} + O(1)
             \\
    &  =  H(\Delta') + O(H(\Delta')^{2/3}+1) 
             \\
    & \le \mu_D(T') + O(\mu_D(T')^{2/3}+1) 
             && \text{(by \thmref{shannon})} \\
    & \le \mu_D(T^*) + O(\mu_D(T^*)^{2/3}+1) 
             && \text{(by \eqrefx{t-triangle})} 
 \end{aligned}
\]
Finally, since we preprocess $\Delta$ using \thmref{ammw07}, the
expected number of comparisons required to answer a query is
\begin{eqnarray*}
  \mu_D(R) 
   & = & H(\Delta) + O(H(\Delta)^{1/2} + 1) \\
   & \le & \mu_D(T^*) +  O(\mu_D(T^*)^{2/3} + 1) 
\end{eqnarray*}
and this completes the proof, and the paper.
\end{proof}

\bibliographystyle{plain}
\bibliography{entropy2}

\end{document}